\documentclass[10pt, conference, letterpaper]{IEEEtran}
\usepackage{amsmath,amssymb}
\usepackage{amsthm}
\usepackage{graphicx}
\usepackage{epsfig}
\usepackage{subfigure}
\usepackage{verbatim}

\newtheorem{theorem}{Theorem}

\newtheorem{corollary}{Corollary}

\newtheorem{lemma}{Lemma}

\newcommand{\expectation}{\ensuremath{\mathbb{E}}}
\newcommand{\Tfixed}{T^{\rm F}}
\newcommand{\hfixed}{h^{\rm F}}
\DeclareMathOperator{\var}{\bf{Var}}

\begin{document}

\title{Secure Content Distribution in Vehicular Networks}

\author{
	\IEEEauthorblockN{Viet T. Nguyen, Jubin Jose, Xinzhou Wu and Tom Richardson}
	\IEEEauthorblockA{Qualcomm Research\\
	Bridgewater, NJ\\
     Email: \{ntienvie, jjose, xinzhouw,  tomr\}@qti.qualcomm.com}}

\maketitle

\begin{abstract}
Dedicated short range communication (DSRC) relies on secure distribution to vehicles of a certificate revocation list (CRL) for enabling security protocols. CRL distribution utilizing vehicle-to-vehicle (V2V) communications is preferred to an infrastructure-only approach. One approach to V2V CRL distribution, using rateless coding at the source and forwarding at vehicle relays is vulnerable to a pollution attack in which a few malicious vehicles forward incorrect packets which then spread through the network leading to denial-of-service. This paper develops a new scheme called Precode-and-Hash that enables efficient packet verification before forwarding thereby preventing the pollution attack. In contrast to rateless codes, it utilizes a fixed low-rate precode and random selection of packets from the set of precoded packets. The fixed precode admits efficient hash verification of all encoded packets. Specifically, hashes are computed for all precoded packets and sent securely using signatures. We analyze the performance of the Precode-and-Hash scheme for a multi-hop line network and provide simulation results for several schemes in a more realistic vehicular model.
\end{abstract}

\section{Introduction}\label{S: Intro}     
Security protocols designed for V2V communication \cite{DSRC} rely on the assumption of periodic distribution to all vehicles of a Certificate Revocation List (CRL) created by a certificate authority  \cite{CRL08,CRL_Kenneth08,CRL_Owen10}. Since requiring every vehicle to have internet connectivity is a significant obstacle for large-scale adoption, it is important to develop an approach that requires only a few vehicles to obtain the CRL from infrastructure and utilizes V2V communication to distribute the CRL to the rest of the vehicles.

CRL distribution in vehicular networks is a typical file distribution problem with stringent security requirements. The approach of distributing packets in a round-robin or random fashion at the source and relaying (vehicles relaying packets to other vehicles) is very inefficient. Similar to the coupon collector problem, the inefficiency arises from the delay and redundancy increase for successive innovative packets. The utilization of an efficient fountain or rateless code at the source solves this issue. Furthermore, the original file can be appended with a signature to ensure that vehicles can detect any file alteration by malicious vehicles. However, an important and challenging security concern remains, which is pollution attack. The attack involves the forwarding of incorrect packets by one or more malicious vehicles. If these malicious packets are forwarded by other non-malicious relays then the incorrect packets will spread very fast and many vehicles will not be able to decode the original file, even when there is very few malicious nodes. Hence, secure content distribution requires that vehicles forward only verified packets.

If a signature is appended to the original file, then a relay can wait until it decodes the entire file to verify all the received packets. All vehicles (except possibly malicious vehicles) do not forward packets until the file is decoded. We refer to such a scheme as \emph{Wait-to-Decode}. An important advantage is that re-encoding is possible after decoding. However, there are drawbacks: (\emph{a}) If the file is large there could be a large initial delay before relays successfully decode and begin to forward. (\emph{b}) Since incorrect packets can be forwarded by malicious vehicles to many vehicles, decoding at these vehicles could fail without error correcting methods.

Another technique to enable packet verification is to include signatures on each encoded packet after applying fountain coding at the source. The advantages of this \emph{Sign-every-Packet} scheme include the ability to verify each individual packet and allow forwarding before decoding the file. However, there are disadvantages: (\emph{a}) The overhead introduced is big and grows linearly with the number of packets in the file. (\emph{b}) The computational complexity to verify each packet is high. (\emph{c}) The relays cannot re-encode after the file is decoded as the signature can be included by only the source. 

\begin{figure}[h]
\begin{center}
\includegraphics[width=9cm]{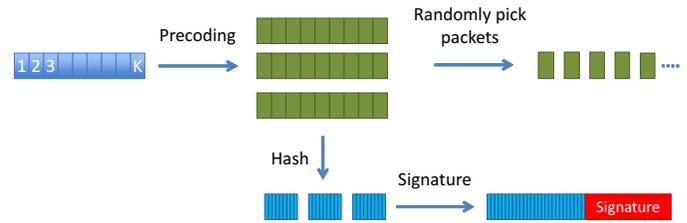}
\caption{Precode-and-Hash Scheme: A new scheme that enables efficient packet verification with low overhead}
\label{fig:precode_and_hash_scheme}
\end{center}
\end{figure}

In this paper, we develop a new scheme, \emph{Precode-and-Hash}, and characterize its performance using both analytical and simulation results. Although existing fountain codes are throughput efficient, adding packet verification into these schemes is challenging due to infinitely many potential coded packets. By using a simple precode (for example, of rate ${1}/{N}$), we limit the number of potential coded packets to $N$ times the file size at the expense of some throughput overhead. We take advantage of this to design an efficient packet verification scheme. Specifically, we separately hash all precoded packets and group them together into one or few hash-information packets. A fountain code is applied to these hash-information packets. The resulting coded hash-information packets are signed by the source and distributed in the network by forwarding, possibly with higher priority in the beginning. This new scheme is captured in Figure \ref{fig:precode_and_hash_scheme}. Once a relay has received these hash-information packets, it can verify each packet by taking its hash and comparing with the hash it received securely (i.e., hash-information packets that are signed). Since hashes are small ($20$ bytes for widely used SHA1 hash) and computationally efficient, packet verification is efficient both in terms of overhead and complexity. More details on Precode-and-Hash scheme can be found in Section \ref{S: Precode-and-hash}.

The main contributions in this paper include the following.
\begin{itemize} 
\item We analyze our scheme for multi-hop line network and prove that a speed-up factor of at least $2$ is possible over wait-to-decode scheme. 
\item We prove the convergence to fluid limit and show that a speed-up factor of $e$ is achievable. Furthermore, we show that this result is tight. 
\item We provide detailed simulations of above mentioned schemes for Boston urban-area model in Section \ref{S: Simulation}. 
\end{itemize}
Both analysis and simulations show that Precode-and-Hash is able to solve the pollution attack problem with limited overhead compared to a non-secure fountain code scheme. Hence, Precode-and-Hash is a highly desirable candidate for secure CRL distribution in vehicular networks.

\subsection{Related Work}
 Network coding - a coding paradigm that allows re-encoding at relays - have the potential to achieve the maximum information flow capacity in a network \cite{ahlswede2000network,ho2006random}. Network coding and its variants are, however, vulnerable to many security threats \cite{Jing08}, one of which is the pollution attack. Several countermeasures have been proposed to cope with this challenge. The general approach is to design hash functions that allows hash of a linearly combined packet (produced by Network Coding) to be computed as a function of the hashes of the combining packets. Some examples include homomorphic hash \cite{KFM04}, secure random checksum \cite{GR06} and the very smooth homomorphic hash \cite{LLD12}. however, the complexity of computing these hashes is still very far from being practical. We, on the other hand, design an alternative delivery scheme that is sub-optimal, compared to network coding, but provides protection against pollution attack at very low computation cost using (any) standard hash functions. 

\section{System Model}
\label{S: Sys Mod}

A source (e.g. roadside unit) has a file (e.g. certificate revocation list) that consists of $k$ packets. The source intends to send the same file to all the vehicles through a wirelessly connected (e.g. DSRC) vehicular network, where each transmission may be loss with probability $\epsilon$. A transmission is successful if the intended transmitter is within a certain range (transmission range) and all other transmitters are beyond a certain range (interference range) from the receiver. All the vehicles move around in a geographical region based on a mobility model.

We consider the following network topologies and mobility models.
\begin{itemize}
\item Multi-hop line network: This network is formed by a set of $d$ static nodes placed along a line. Node $0$ is the source while nodes $1,2,\ldots,d-1$ are relays and destinations. Node $i$ is in the range of nodes $i-1$ and $i+1$. For example, node $0$ is in range of node $1$, node $1$ is in range of nodes $0$ and $2$, and so on.
\item Boston urban-area model: This model is expected to capture real-world topology and mobility. We model majority of the roads in Boston-Cambridge area and simulate traffic on this model. Each vehicle's movement at every intersection follows a Markov chain $P =  [P_{ij}]$, where $P_{ij}$ is the probability to switch from directed road segment $i$ to directed road segment $j$. We calibrate $P$ using the real daily traffic volume data reported for each major road segment in this area \cite{BostonModel}. More details on this model are given in Section  \ref{S: Simulation}. 
\end{itemize}

\section{Precode-And-Hash} 
\label{S: Precode-and-hash}
 The content distribution scheme is depicted in Figure (\ref{fig:precode_and_hash_scheme}). It consists of two main components: 
\begin{itemize}
\item Precode: A fixed-rate (typically, low-rate) erasure code is applied as a precode. Let the rate of this precode be ${1}/{M}$ (for example, $1/3$ in Figure \ref{fig:precode_and_hash_scheme}). If a file consists of $k$ packets, there will be $Mk$ coded packets in total. We assume an optimal erasure code so the file can be successfully decoded if a receiver has \emph{any} $k$ out of $Mk$ coded packets.\footnote{A practical code could introduce a small overhead, but it is often negligible with large file size.} Since the decoding is successful whenever a vehicle receives \emph{any} $k$ coded packets, the \emph{rareness} issue is solved. More precisely, the $\log k$ overhead arising in the coupon collector problem is reduced to a constant overhead, which decreases with lower code rate (or higher number of coded packets).
\item Hash: Hashes are computed by the source for each coded packet using a sufficiently hard hash function. 
These hash packets must be distributed to the whole network before data packets are sent out. Since the number of hash-information packets is small, we expect the overhead for this initial distribution of hashes to be small. These hash make it possible for other nodes to quickly check if a coded packet is polluted. 
\end{itemize}


A node, depending on its role, performs the following operations.
\begin{itemize}
\item \emph{A source} chooses a random coded packet and broadcasts to its neighbors.
\item \emph{A relay} has to be both a receiver and a transmitter. As a transmitter, it chooses a random coded packet in the set of the coded packets that it receives and broadcast it to its neighbors.  As a receiver, it discards all polluted packets and when its buffer contains at least $k$ distinct coded packets, it can reconstruct the original file. Then, that relay can apply the precode to the original file and acts as a secondary source.
\end{itemize}

\section{Analysis of Delay Performance}
\label{SS: analysis}

We are interested in the delay performance (i.e., the time required for all the nodes in the network to successfully receive the data file) for the precode-and-hash scheme. Since quantifying this delay seems intractable for finite file size, the focus is on the asymptotic delay, i.e., the limit of the ratio between the distribution delay and file size when the file size goes to infinity. This value is a good indicator of the overhead introduced by the Precode in the precode-and-hash scheme when the file size is sufficiently large.

\subsection{Discreet Analysis}
\label{SS:discreet analysis}

At time $t$ some nodes have the file and the remainder have
a subset $H_i(t)$ of the coded packets. WLOG, we assume $H_i(t) \subset H_{i-1}(t)$ and if node $i$ has the file then so does node $i-1.$ It is then sufficient to describe the system at time $t$ by an infinite sequence of decreasing number $\mathcal{H}(t) = \{|H_{0}(t)|, |H_{1}(t)|,...\}$. It is easily seen that such a system is indeed Markovian, where $|H_i(t+1)| = |H_i(t)| + 1$ with probability $1 - \epsilon|H_i(t)|/|H_{i-1}(t)|$ and remains the same otherwise . It's also worth notice that when $ |H_i(t)| + 1 = k$, that value is substituted by $Mk$ and remains at this value as node $i$ becomes a secondary source. 
 Let $T_{n}$ be the time $t$ such that $H_n(t)$ first become $Mk$, the first observation is encapsulated in the following Lemma.

\begin{lemma}
\label{P:2 hops partial collection}
$\mathbb{P}(|H_{n+1}(T_{n})|/(k-1) \geq  0.5)$ converges to $1$ as $k \rightarrow \infty$. 
\end{lemma}

 This Lemma shows that by the time node $n-1$ become a secondary source, node $n$ has already at least half the packets needed that are needed for decode. Hence the delivery time from node $n-1$ to node $n$ decreases even further, compared to the one hop delay from node $0$ to node $1$. The next Theorem characterizes this speed up. 
\begin{theorem}
\label{P: n hops total collection}
Assuming the relays initially have no packet, we have,
\[\mathbb{P}\left (\frac{(1-\epsilon)T_{n}}{k} \leq M\log\frac{M}{M-1}+(n-1)M\log \frac{2M-1}{2M-2}\right ) \rightarrow 1\]
 as $k \rightarrow \infty$. Hence,
\[\frac{(1-\epsilon)T_{n}}{k} \leq 1+ \frac{n-1}{2}\]
 holds with large probability for large enough $k, N$ and 
 \[\frac{(1-\epsilon)T_{n}}{nk} \leq M\log \frac{2M-1}{2M-2}\]
 holds with large probability for large enough $k, n .$
\end{theorem}

So, the additional file decoding delay for multi-hop line network compared to single-hop grows as $(n-1)/2$ ($n$ is the number of hops) with Precode-and-Hash scheme, given the precode rate $1/M$ is low enough. In contrast, the file decoding delay for Wait-to-Decode scheme grows as $n-1$.

\subsection{Fluid limit Analysis} \label{SS:fluid}
As the file size $k\rightarrow\infty$, we can prove that the file distribution convergences sharply to  the
fluid limit (e.g. by Wormald's approach) for any finite time and length. The fluid limit model is as follows.

Let $\frac{1}{k}|H_i(t)|=h_i(t) < 1.$ The function $h_i(t)$ is non-decreasing in $t.$ and continuous except at time $\mathcal{T}_{i}$ when it reaches $1.$ At that point the node $i$ is able to decode the file and then for $t > \mathcal{T}_i$ we have $h_i(t) = M.$
For $i > 0$ and $t \not\in\{ \mathcal{T}_{i-1},\mathcal{T}_i \}$ we have $\frac{d}{dt} h_i(t) = 1 -\frac{h_i(t)}{h_{i-1}(t)}.$
For $t \in\{ \mathcal{T}_{i-1},\mathcal{T}_i \}$ we have $
\frac{d}{dt} h_i(t) = 1 -\frac{M}{h_{i-1}(t)}.$ 
Then, we can further improve the asymptotic bound on distribution delay in $n$-hop chain in Proposition \ref{P: n hops total collection}. Furthermore, the later bound is tight. The main result of this analysis is given next.

\begin{theorem}\label{T:fluid bound}
Assuming the initial condition $h_i(0)=0, i>0,$ we have,
for $M \ge 2,$
\[
\lim_{n \rightarrow \infty} \frac{\mathcal{T}_n}{n} = \Tfixed(M)
\]
where $\Tfixed(M)$ is the unique  $x \in [0,1]$ solving
\(
-\ln x = 1 - \frac{x}{M}\,.
\)
\end{theorem}

Note that the file size is scaled down to $1$, hence we should have $\mathcal{T}_i \approx \frac{(1-\epsilon)T_i}{k}$ (the term $(1-\epsilon)$ accounts for lost packets). The above result can be rewritten as follows.
\begin{corollary}
For very large $k, n $, the followings hold with large probability
\[\frac{(1-\epsilon)T_{n}}{kn} = \Tfixed(M) + o(1)\]
\end{corollary}  
 Compare to Proposition \ref{P: n hops total collection}, the asymptotic bound decreases by a factor at least $e/2$ and can be up to $1.74$ (when $M=2$).

\section{Simulation Results \& Comparison}
\label{S: Simulation} 

\subsection{Scenarios \& Parameters}\label{SS: Sim setting}

The primary application of interest is CRL (certificate revocation list) distribution in vehicular networks. Each vehicle is securely identified by its pseudonym ID, which is issued by the Certification Authority (CA). For privacy reasons (e.g., preventing tracking), each vehicle is expected to change its pseudonym periodically (e.g., once every $10$ minutes). Each vehicle
is equipped with sufficient pseudonyms for long term operation.
If the CA distrusts a vehicle, then it has to revoke all pseudonyms associated with that vehicle. This information is distributed to all other vehicles in a CRL file. CRL information corresponding to each vehicle is expected to be around 40 kB or higher. Hence, assuming tens of vehicles in the list, CRL file size is expected to be few MB. In our simulation, we assume that the CRL file size is $1$ MB which is split into $1000$ packets of $1000$ bytes each. 



We model V2V wireless communication using the following abstractions. The physical layer is modeled as an erasure channel with a packet erasure probability of $\epsilon=0.05$. The medium access layer (MAC) is modeled by slotted carrier sensing (CSMA/CA). A transmission is successful if the intended transmitter is within a certain range (transmission range of $200$ m) and all other transmitters are beyond a certain range (interference range of $300$ m) from the receiver.A vehicle can transmit $20$ packets each time slot if it is elected by CSMA/CA. 

\begin{figure}
\begin{center}
\subfigure[Satellite map of Boston urban-area]{
\hspace{0.3cm}\includegraphics[width=0.4\textwidth]{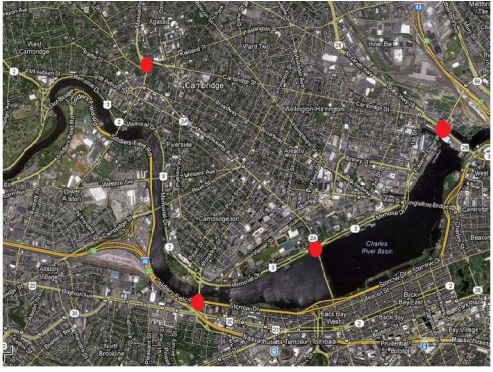}
}
\subfigure[Road map of Boston urban-area]{
\includegraphics[width=0.46\textwidth]{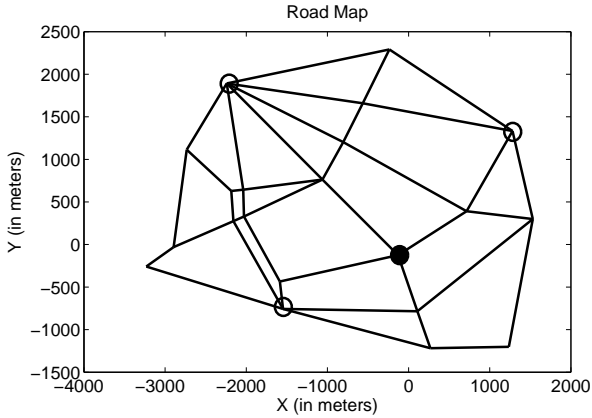}
}
\end{center}
\caption{Map of Boston urban-area}
\label{BostonMap}
\end{figure}

 We consider $236$ vehicles moving with average velocity of $20 m/s$, where each vehicle has a time-invariant velocity randomly (uniform) selected from $[15,25]m/s$. There are four sources at fixed locations as shown in Figure \ref{BostonMap}. The sources continuously seed the file for a time period equal to $2-5$ times the time to seed the given file size. For a fair comparison between different schemes, the seeding time is defined with respect to the original file size instead of the actual number of source packets. The seeding rate, i.e., the transmission rate of sources, is chosen to be $60$ packets per second.

\subsection{Distribution Schemes and Security Overhead}\label{SS: Security Overhead}
We recall that the scheme developed in this paper for secure content distribution is Precode-and-Hash.

 The hash-information packets are encoded and distributed using the Sign-every-Packet scheme. Each arriving packet is classified as, authentic, polluted or unknown (hash not received). Each vehicle is allowed to forward either a coded hash-information packet or an authentic coded data packet. Initially, a vehicle transmits coded hash-information packets for a fixed duration depending on the total size of the hashes. After this initial phase, a vehicle sends coded hash-information packets with probability $0.2$ and coded data packet with probability $0.8.$. 

The intuition is that, a vehicle initially has more coded hash packets than coded data packets, so it will send coded hash packets first to ensure that the hashes are distributed ``ahead'' of the data. After that, we reserve enough capacity to distribute the hashes by choosing the right hash forwarding probability ($0.2$), which is fine-tuned by experiments.



 As described in Section \ref{S: Intro}, two other schemes that could provide protection against pollution attack are Wait-to-Decode and Sign-every-Packet. 

  The security overhead associated with each scheme is:
\begin{itemize}
\item Wait-to-Decode: A signature is appended to the whole CRL file. The signature size is $256$ bytes for RSA-SHA2. Hence we can neglect this overhead.
\item Sign-every-Packet: Each packet reserves $256$ bytes for signature. The number of packets increases to  $\lceil10^6/(1000 -256)\rceil = 1344$. The overhead (in terms of packets) is about $34\%$.
\item Precode-And-Hash: Each hash is $20$ bytes in size.  As each hash-information packets is individually signed after the fountain coded is applied, each packet contains at most $\lfloor(1000 - 256) / 20\rfloor = 37$ hashes and there are at least $1000M/37 \approx 27M$ such packets.  The overhead is $8 \%$ with $M=3$, $11 \%$ with $M=4$ and $13 \%$ with $M=5$. In our simulations, $M=3.$
\item Genie Precode: This schemes use a rateless Precode instead of a finite rate one. It is considered here for comparison of data (excluding hash) delivery delay. 
\end{itemize}

%

\begin{figure}[h]
\begin{center}
\includegraphics[width=9cm]{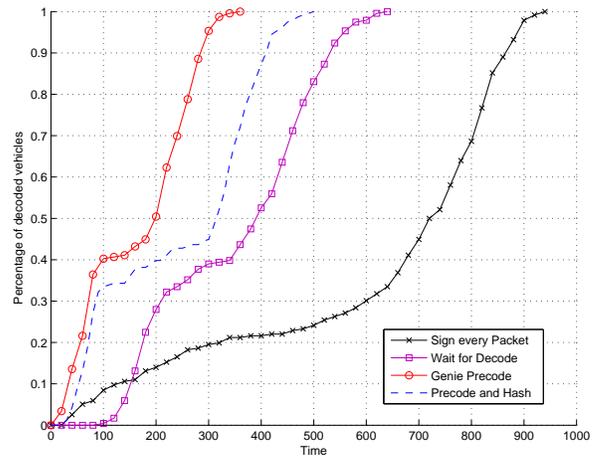} 
\caption{Performance comparison: Boston urban-area model, seeding time enough to transmit $5$ times the file size}
\label{fig: Boston 1000 2}
\end{center}
\end{figure}

\begin{figure}[h]
\begin{center}
\includegraphics[width=9cm]{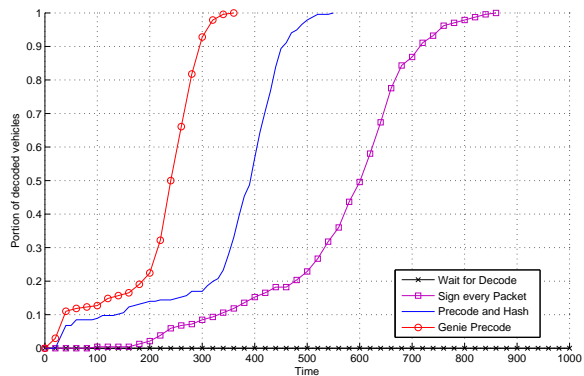} 
\caption{Performance comparison: Boston urban-area model, seeding time enough to transmit $2$ times the file size}
\label{fig: Boston 1000 5}
\end{center}
\end{figure}
\subsection{Simulation Results}\label{SS: Result}
There are four sources that seed the file. The seeding period equal to $5$ times the duration to seed the file size (Figure \ref{fig: Boston 1000 5}). Next, to understand the performance of different schemes with smaller seeding time, we reduce the seeding period to $2$ (Figure \ref{fig: Boston 1000 2}). In all simulations, we plot the fraction of vehicles that has decoded the file as a function of time. Based on the result, the following remarks is made.
\begin{itemize}
\item Precode-and-Hash out-performs both Sign-every-Packet and Wait-for-Decode in the scenarios considered. Wait-for-Decode has very large file distribution delays. 
\item Precode-and-Hash is $20 \%$ to $30 \%$ worse than Genie Precode. This is expected since we reserve $20\%$ capacity for hash distribution.  
\item In Figure \ref{fig: Boston 1000 5}, Wait-for-Decode distribution and Sign-every-Packet has a crossover due to smaller file size for Wait-for-Decode (recall that files size in Sign-every-Packet is $34\%$ larger than other schemes). However, Sign-every-Packet is still better compared to Wait-for-Decode for majority of the vehicles.
\item Both Wait-for-Decode and Sign-every-Packet are very sensitive to seeding parameters while Precode-and-Hash is more robust. In Wait-for-Decode, if no relay decodes file during the seeding period, no one can start re-seed the file after that. Hence, the file is not distributed at all. In Sign-every-Packet, as relays can only forward coded packets that they receive, only a small subset of coded packets can be distributed if the seeding time is not long enough. The decoding is sub-optimal as the innovative information is limited.  In Figure \ref{fig: Boston 1000 2}, when the seeding time is reduced to $2$, (a) Wait-for-Decode does not spread the file, (b) the file distribution time of Sign-every-Packet increases considerably while (c) the file distribution time of Precode and Hash only increases by $10\%$.  
\end{itemize}

\subsection{Practical Discussion}
 In the Precode and Hash scheme, the first practical choice one has to make is the Precode rate. This involves balancing the data delivery time and the capacity reserved for hash distribution as discussed in \ref{SS: Security Overhead}. Note that the capacity reserved for hash distribution is much larger than the hash size itself since the hash distribution is much less efficient than data distribution

 Another observation is that the speed-up in
file propagation, which can be up to $e$ comes at the cost of an infinite ratio of wasteful transmissions (i.e, the transmitted packet is already in the receiver buffer). As in wireless network environment, extraneous transmissions increases the contention and interference (and hence increases packet loss probability) significantly, this problem may become very serious. 

One way to avoid this issue is to let node forward packets with a probability proportioned to its buffer size. An analysis of this proportioned forwarding scheme is provided in Appendix \ref{A:proportion forwarding}. The result shows that we still achieve a speed up factor at least $1.66$ by using much less forwarding.
\section{Conclusion} 
The Precode-and-Hash scheme developed in this paper provides protection against pollution attacks with a minimal increase in file delivery time. Hash verification is very efficient compared to signature verification. Since the file distribution delay is limited by physical mobility, the increase in delay of file distribution due to the use of finite-rate precode is minimal. However, the hash distribution delay can be significant. We expect that the net delay can be further improved with more optimized hash distribution. In summary, the Precode-and-Hash scheme is an attractive candidate for secure content distribution in vehicular networks.

\bibliography{ref}{}

\begin{thebibliography}{10}
\providecommand{\url}[1]{#1}
\csname url@samestyle\endcsname
\providecommand{\newblock}{\relax}
\providecommand{\bibinfo}[2]{#2}
\providecommand{\BIBentrySTDinterwordspacing}{\spaceskip=0pt\relax}
\providecommand{\BIBentryALTinterwordstretchfactor}{4}
\providecommand{\BIBentryALTinterwordspacing}{\spaceskip=\fontdimen2\font plus
\BIBentryALTinterwordstretchfactor\fontdimen3\font minus
  \fontdimen4\font\relax}
\providecommand{\BIBforeignlanguage}[2]{{%
\expandafter\ifx\csname l@#1\endcsname\relax
\typeout{** WARNING: IEEEtran.bst: No hyphenation pattern has been}%
\typeout{** loaded for the language `#1'. Using the pattern for}%
\typeout{** the default language instead.}%
\else
\language=\csname l@#1\endcsname
\fi
#2}}
\providecommand{\BIBdecl}{\relax}
\BIBdecl

\bibitem{DSRC}
``Standard specifications for telecommunications and information exchange
  between roadside and vehicle systems - {5 GHz} band dedicated short range
  communications ({DSRC}) medium access control {(MAC}) and physical layer
  ({PHY}) specifications,'' Sept. 2003.

\bibitem{CRL08}
P.~Papadimitratos, G.~Mezzour, and J.-P. Hubaux, ``Certificate revocation list
  distribution in vehicular communication system,'' \emph{Proc. ACM
  International Workshop on Vehicular Internetworking, San Francisco, CA},
  2008.

\bibitem{CRL_Kenneth08}
K.~P. Laberteaux, J.~J. Haas, and Y.-C. Hu, ``Security certificate revocation
  list distribution for {VANET},'' \emph{Proc. ACM International Workshop on
  Vehicular Internetworking, San Francisco, CA}, 2008.

\bibitem{CRL_Owen10}
M.~Nowatkowski and H.~Owen, ``Scalable certificate revocation list distribution
  in vehicular ad hoc networks,'' \emph{Proc. SWiM'10, Miami, FL,}, 2010.

\bibitem{ahlswede2000network}
R.~Ahlswede, N.~Cai, S.-Y. Li, and R.~W. Yeung, ``Network information flow,''
  \emph{IEEE Transactions on Information Theory}, vol.~46, no.~4, pp.
  1204--1216, 2000.

\bibitem{ho2006random}
T.~Ho, M.~M{\'e}dard, R.~Koetter, D.~R. Karger, M.~Effros, J.~Shi, and
  B.~Leong, ``A random linear network coding approach to multicast,''
  \emph{IEEE Transactions on Information Theory}, vol.~52, no.~10, pp.
  4413--4430, 2006.

\bibitem{Jing08}
J.~Dong, R.~Curtmola, R.~Sethi, and C.~Nita-Rotaru, ``Toward secure network
  coding in wireless networks: Threats and challenges,'' in \emph{4th Workshop
  on Secure Network Protocols (NPSec)}, Oct 2008, pp. 33--38.

\bibitem{KFM04}
M.~Krohn, M.~Freedman, and D.~Mazieres, ``On-the-fly verification of rateless
  erasure codes for efficient content distribution,'' in \emph{IEEE Symposium
  on Security and Privacy}, May 2004, pp. 226--240.

\bibitem{GR06}
C.~Gkantsidis and P.~Rodriguez, ``Cooperative security for network coding file
  distribution,'' in \emph{IEEE International Conference on Computer
  Communications (INFOCOM)}, April 2006, pp. 1--13.

\bibitem{LLD12}
Q.~Li, J.-S. Lui, and D.-M. Chiu, ``On the security and efficiency of content
  distribution via network coding,'' \emph{Dependable and Secure Computing,
  IEEE Transactions on}, vol.~9, no.~2, pp. 211--221, March 2012.

\bibitem{BostonModel}
MassDOT, ``Internet: www.mhd.state.ma.us,'' 2012.

\end{thebibliography}
\bibliographystyle{IEEEtran}

\appendix

\subsection{Proof of Proposition \ref{P:one-hop}}
From Lemma \ref{L: Geo step delay} we have
\begin{align*}
&\mathbb{E}\left [\frac{t_{1,k}}{k}\right ] = \frac{1}{k}\sum_{i=0}^{k-1}\mathbb{E}\left [t_{1,i+1} - t_{1,i}\right ]
= \sum_{i=0}^{k-1}\frac{N}{(kN-i)(1-\epsilon)} 
\\&=\frac{N}{1-\epsilon}  \frac{1}{kN} \sum_{i=0}^{k-1}\frac{1}{(1-\frac{i}{kN})}
= \frac{N}{1-\epsilon}\int_{0}^{1/N}\frac{1}{1-t}dt + O(k^{-1})\\
&=\frac{N}{1-\epsilon}\log(\frac{N}{N - 1}) + O(k^{-1}). 
\end{align*}
To show convergence in probability we show that the variance 
of $\frac{t_{1,k}}{k}$ tends to $0.$
Now, since $t_{1,i+1} - t_{1,i}$ are independent random variables,
\begin{align*}
\var(t_{1,k})&=\sum_{i=0}^{k-1}\var(t_{1,i+1}- t_{1,i})=\sum_{i=0}^{k-1}p_i^{-2}(1-p_i)\\
\le \sum_{i=0}^{k-1}p_i^{-2} &= 
\frac{kN}{(1-\epsilon)^2}\frac{1}{kN}\sum_{i=0}^{k-1}\frac{1}{(1-\frac{i}{kN})^2}\\
			&= \frac{kN}{(1-\epsilon)^{2}}\int_{0}^{\frac{1}{N}}\hspace{-.2cm}\frac{1}{(1-t)^{2}}dt + o(1)\\
			& = \frac{kN}{(1-\epsilon)^2 (N-1)}  + o(1).			
\end{align*}
It now follows, e.g. from Chebyshev's inequality, that $\frac{t_{1,k}}{k}$ converges in probability to $\frac{N}{1-\epsilon}\log(\frac{N}{N - 1})$. We use the fact that $\log(1+x)\le x$ for $x>-1$ for the bound.

\subsection{Proof of Subsection \ref{SS:discreet analysis}} \label{A: proof of Multihop}  
The following lemmas are needed in the proof of Proposition \ref{P:2 hops partial collection}.

First, we need to characterize the distribution of $I_{n,i}$.
Let us introduce the notation $\partial t_{n,i} = t_{n,i+1}- t_{n,i}$
and ${\cal T}_{n,k}=(t_{n,1},\ldots,t_{n,k}).$

\begin{lemma}\label{L: appearance prob.}
Conditioned on ${\cal T}_{n-1,k}$, the r.v.s $I_{n,i},\ i=1,\ldots,k$, are $k$ pairwise negatively dependent $\{0,1\}$-values r.v.s with conditional distribution
\begin{align}
\nonumber &\mathbb{P}(I_{n,i}=0\,|\,{\cal T}_{n-1,k})
=\prod_{j=i}^{k-1}\left (1-\frac{1-\epsilon}{j}\right )^{\partial t_{n-1,j}}.\label{E: appearance prob.}
\end{align}
\end{lemma}
\begin{proof} 
By definition, packet  $i$ is available at node $n-1$ starting from time $t_{n-1,i}.$
For each subsequent time the probability that packet $i$ is not transmitted to the second hop is $1-\frac{1-\epsilon}{j},$ hence the stated distribution holds.

First we note that for any finite random variables $X,Y$ 
\(
\mathbb{E} (X Y)-\mathbb{E}( X) \mathbb{E} (Y)
=\mathbb{E} (1-X)(1-Y) -\mathbb{E} (1-X)\mathbb{E}(1-Y)
\)
so we show negative dependence of  $\{1-I_{n,i},\ i=1,\ldots,k\}.$
Let $i_{1}<i_{2}$, then
\begin{align*}
&\mathbb{E}[(1-I_{n,i_1})(1-I_{n,i_2})\,|\,{\cal T}_{n-1,k}]\\
&=\mathbb{P}(I_{n,i_1}=0 \textnormal{ and } I_{n,i_2}=0\,|\,{\cal T}_{n-1,k}).
\end{align*}
The above expression is the probability that neither $i_{1}$ nor $i_{2}$ are transmitted to the second hop at any time from $1$ to $t_{n-1,k}$.
Packet $i_1$ is available from time $i_1$
and both packets $i_1$ and $i_2$ are available from time $i_2.$
After this time, the probability that neither is successfully transmitted in a given slot while
$j$ packets are available is 
\(
1-2\frac{1-\epsilon}{j} \le 
(1-\frac{1-\epsilon}{j})^2\,.
\)
Thus we obtain
\begin{align*}
&\mathbb{P}(I_{i_{n,1}}=0 \textnormal{ and } I_{i_{n,2}}=0\,|\,{\cal T}_{n-1,k})\\
&=\prod_{j=i_{1}}^{i_{2}-1}\left (1-\frac{1-\epsilon}{j}\right )^{\partial t_{n-1,j}}
\prod_{j=i_{2}}^{k-1}\left (1-2\frac{1-\epsilon}{j}\right )^{\partial t_{n-1,j}}\\
&\leq\prod_{j=i_{1}}^{i_{2}-1}\left (1-\frac{1-\epsilon}{j}\right )^{\partial t_{n-1,j}}
\prod_{j=i_{2}}^{k-1}\left (1-\frac{1-\epsilon}{j}\right )^{2(\partial t_{n-1,j})}\\
&=\prod_{j=i_{1}}^{k-1}\left (1-\frac{1-\epsilon}{j}\right )^{\partial t_{n-1,j}}
\prod_{j=i_{2}}^{k-1}\left (1-\frac{1-\epsilon}{j}\right )^{\partial t_{n-1,j}}\\
&=\mathbb{P}(I_{n,i_{1}}=0 \,|\,{\cal T}_{n-1,k})
\mathbb{P}(I_{n,i_{2}}=0 \,|\,{\cal T}_{n-1,k}). \mbox{ a.s. }
\end{align*}
This completes the proof.
\end{proof}

The next step is to characterize the distribution of $t_{n-1,1},\ldots,t_{n-1,k}$. 

\begin{lemma}\label{L: k hops arrival time}
For any $k$, there exist $k-1$ independent geometric r.v.s $\tau_{1},\ldots,\tau_{k-1}$ of parameters $\frac{(1-\epsilon)(kN-1)}{kN},\ldots,\frac{(1-\epsilon)(kN-k+1)}{kN} $ respectively such that $t_{n,j+1}-t_{n,j}\geq \tau_{j}$ a.s. for $j=1,\ldots,k-1$. 
\end{lemma}
\begin{proof}Let $\mathcal{E}:=\{1,\ldots,k\}^{n}$. The $k$-tuple comprised of the number of packets collected by the $k$ hops at each time slot form a Markov chain with state space $\mathcal{E}$. The fact that $t_{k,j+1}-t_{k,j},\, j=1,\ldots,M-1$ are independent comes from the strong Markov property and from the fact that $t_{n,j},j=1,\ldots,k$ are $k$ increasing stopping times of this Markov chain. To complete the proof, we now need to show that
\begin{align*}
\mathbb{P}(t_{n,j+1}-t_{n,j}>l) \geq \left (1-\frac{(1-\epsilon)(kN-j)}{kN}\right )^{l}.  
\end{align*}  
In particular, let $H_{t_{n,j}}, H_{t_{n,j}+1}, \ldots, H_{t_{n,j}+l-1}$ be the number of coded packets at the $n-1$ hop at time $t_{n,j},\ldots,t_{n,j}+l-1$. Note that the set of packets received at node $n$ at any time must be a subset of the set of packets received at node $n-1$. At slot $t_{n,j}+i$, a new packet is received at node $n$ iff node $n-1$ chooses one of the $H_{t_{n,j}+i}-j$ that node $n$ does not have and the transmission is successful. The probability of this event is $\frac{(H_{t_{n,j}+i}-j)(1-\epsilon)}{H_{t_{n,j}+i}}$. So, 
\begin{align*}
&\mathbb{P}(t_{n,j+1}-t_{n,j}>l\,|H_{t_{n,j}}, H_{t_{n,j}+1}, \ldots, H_{t_{n,j}+l-1})\\
& = \prod_{i=0}^{l-1}\left (1-\frac{(H_{t_{n,j}+i}-j)(1-\epsilon)}{H_{t_{n,j}+i}}\right ).
\end{align*} 
On the other hand $H_{t_{n,j}+i} \leq kN$ a.s. for $i=\overline{0,l-1}$ a.s. Hence,
\begin{align*}
&\mathbb{P}(t_{n,j+1}-t_{n,j}>l)\\
&= \mathbb{E}[\mathbb{P}(t_{n,j+1}-t_{n,j}>l\,|H_{t_{n,j}}, H_{t_{n,j}+1}, \ldots, H_{t_{n,j}+l-1})]\\
&\geq \left (1-\frac{(kN-j)(1-\epsilon)}{kN}\right )^{l}.
\end{align*}
This completes the proof.
\end{proof}

Two r.v.s $X$ and $Y$ are negatively dependent iff $\mathbb{E}[XY] \leq \mathbb{E}[X]\mathbb{E}[Y]$.
\begin{lemma}\label{L: Negative dependence condition}
Two r.v.s $X, Y$ are negatively dependent if $c-X, c-Y$ are negatively dependent, where $c$ is any real number.
\end{lemma}
\begin{proof} Since $c-X, c-Y$ are negatively dependent,  
\begin{align*}
\mathbb{E}[(c-X)(c-Y)] \leq \mathbb{E}[c-X]\mathbb{E}[c-Y] .
\end{align*}
Expanding the left hand side and the right hand side gives,
\begin{align*}
\mathbb{E}[c-X]\mathbb{E}[c-Y]&=(c - \mathbb{E}[X])(c - \mathbb{E}[Y])\\
							 &= c^{2} - c\mathbb{E}[X] - c\mathbb{E}[Y] + \mathbb{E}[X]\mathbb{E}[Y]\\
\mathbb{E}[(c - X)(c - Y)] &= \mathbb{E}[c^{2} - cX -cY + XY]\\
						  &= c^{2} -c\mathbb{E}[X] -c\mathbb{E}[Y] + \mathbb{E}[XY].
\end{align*}
This implies 
\begin{align*}
\mathbb{E}[XY] \leq \mathbb{E}[X]\mathbb{E}[Y].
\end{align*}
Hence, $X, Y$ are negatively dependent.
\end{proof}

The following lemma bounds the variance of the sum of mutually negatively dependent random variables (r.v.s).

\begin{lemma}\label{L: Var bound}
Let $X_{1},\ldots,X_{n}$ be $n$ mutually negatively dependent r.v.s. Then,
\begin{align*}
Var\left (\sum_{i=1}^{n}X_{i}\right ) \leq \sum_{i=1}^{n}Var(X_{i}).
\end{align*}
\end{lemma}
\begin{proof}First of all,
\begin{align*}
\mathbb{E}\left [\left (\sum_{i=1}^{n}X_{i}\right )^{2}\right]&= \mathbb{E}\left [\sum_{i=1}^{n}X^{2}_{i}+2\sum_{1\leq i< j\leq n}X_{i}X_{j}\right]\\
&=\sum_{i=1}^{n}\mathbb{E}\left [X^{2}_{i}\right ]+2\sum_{1\leq i< j\leq n}\mathbb{E}\left [X_{i}X_{j}\right]\\
&\leq \sum_{i=1}^{n}\mathbb{E}\left [X^{2}_{i}\right ]+2\sum_{1\leq i< j\leq n}\mathbb{E}\left [X_{i}\right ]\mathbb{E}\left [X_{j}\right].
\end{align*}
Moreover,
\begin{align*}
\left (\mathbb{E}\left [\left (\sum_{i=1}^{n}X_{i}\right )\right]\right )^{2}&= \sum_{i=1}^{n}\left (\mathbb{E}\left [X_{i}\right ]\right )^{2}+2\sum_{1\leq i< j\leq n}\mathbb{E}\left [X_{i}\right ]\mathbb{E}\left [X_{j}\right].\\
\end{align*}
This leads directly to the result.
\end{proof}

\emph{Now we can process to the main proof of Proposition \ref{P:2 hops partial collection}.}

\begin{proof}
We start by computing $\mathbb{E}[N_{n}]$. By Lemma \ref{L: appearance prob.},
\begin{align*}
&\mathbb{E}[N_{n}]=\sum_{i=1}^{k-1}\mathbb{E}[I_{n,i}]\\
			 &=(k-1)-\sum_{i=1}^{k-1}\mathbb{E}\left [\prod_{j=i}^{k-1}\left (1-\frac{1-\epsilon}{j}\right )^{t_{n-1,j+1}-t_{n-1,j}}\right ].
\end{align*}
By Lemma \ref{L: k hops arrival time}, $t_{n-1,j+1}-t_{n-1,j},\, j=1,\ldots, k-1$ are bounded below a.s. by $k-1$ independent geometric r.v.s $\tau_{j},\, j=1,\ldots, k-1$ with parameters $p_j=(1-\epsilon)(kN-j)/kN,\, j=1,\ldots,k-1$ respectively. Hence,
\begin{align*}
&\mathbb{E}\left [\prod_{j=i}^{k-1}\left (1-\frac{1-\epsilon}{j}\right )^{t_{n-1,j+1}-t_{n-1,j}}\right ]\\
&=\prod_{j=i}^{k-1}\mathbb{E}\left [\left (1-\frac{1-\epsilon}{j}\right )^{t_{n-1,j+1}-t_{n-1,j}}\right ]\\
&\leq \prod_{j=i}^{k-1}\mathbb{E}\left [\left (1-\frac{1-\epsilon}{j}\right )^{\tau_{j}}\right ]
\end{align*}
Now for any positive constant $c\le 1$ we have $\mathbb{E}(c^{\tau_j}) =\sum_{k=1}^{\infty}c^{k}(1-p_j)^{k-1}p_j =\frac{c p_j}{1-c+c p_j}.$
Setting $c=1-\frac{1-\epsilon}{j}$ we have $c p_j < 1-\epsilon$ so we obtain
$\frac{c p_j}{1-c+c p_j} \le \frac{1}{\frac{1}{j}+1} =\frac{j}{j+1}.$
Continuing, we now have
\begin{align*}
\mathbb{E}[N_{n}]
&\geq k-1-\left (\sum_{i=1}^{k-1}\prod_{j=i}^{k-1}\frac{j}{j+1}\right ) = k-1-\left (\sum_{i=1}^{k-1}\frac{i}{k}\right )\\
& =k-1 - \dfrac{(k-1)k}{2k} =\frac{k-1}{2}.
\end{align*}

Next, we bound $\var(N_{n})$. Using the fact that $I_{n,i},\, i=1,\ldots,k-1$ have negative pair-wise dependence (Lemma \ref{L: appearance prob.}) conditioned on $t_{i},\, i=1, \ldots, k-1$, we obtain 
\begin{align*}
\var(N_{n})&=\mathbb{E}[\var(N_{n}|t_{n-1,1},\ldots,t_{n-1,k-1})]\\
	  &\leq \mathbb{E}\left [\sum_{i=1}^{k}\var(I_{n,i}|t_{n-1,1},\ldots,t_{n-1,k-1})\right ]\\
	  &=\sum_{i=1}^{k}\var(I_{n,i}).
\end{align*}
Since $I_{n,i}$ are $\{0,1\}$ r.v.s, $\var(I_{n,i})\leq 1$, $\var(N_{n})\leq k$. Then, by Chebyshev's inequality, we have
\begin{align*}
&\mathbb{P}\left (\frac{1}{2}-\frac{N_{n}}{(k-1)} >\delta\right )=\mathbb{P}\left (\frac{(k-1)}{2}-N_{n} >\delta (k-1)\right )\\
&\leq \mathbb{P}\left (\mathbb{E}[N_{n}]-N_{n} >\delta (k-1)\right )
\\&\leq \mathbb{P}\left (|\mathbb{E}[N_{n}]-N_{n} |>\delta (k-1)\right )\\
&\leq \frac{\var(N_{n})}{\delta^{2}(k-1)^{2}} \leq \frac{k}{\delta^{2}(k-1)^{2}}.
\end{align*}
This completes the proof.
\end{proof}

\emph{Proof of Proposition \ref{P: n hops total collection}}
\begin{proof}
We use induction on $n$. Note that the $n^{th}$ hop is the $n+1^{th}$ node. The base case $n=1$ is covered by Subsection \ref{SS: One hop}. Suppose that the result holds for $n-1$. We can write $T_{n}=T_{n-1}+\Delta $, where $T_{n-1}$ is the time it takes for node $n$ to collect $k$ coded packets. By the induction hypothesis, 
\[\mathbb{P}\left (\frac{(1-\epsilon)T_{n-1}}{k} \leq N\log\frac{N}{N-1}+(n-2)N\log \frac{2N-1}{2N-2}\right )\] converges to $1$ as $k$ goes to $\infty$. Moreover, by Proposition \ref{P:2 hops partial collection}, $\mathbb{P}(N_{n}/(k-1)\ge0.5)$ converges to 1. So, by an argument similar to that in Subsection \ref{SS: One hop}, $(1-\epsilon)\Delta/k$ converges to 
\begin{align*}
&\frac{\sum_{i=N_{n}+1}^{k}\frac{kN}{kN-i}}{k}=N\int_{N_{n}/kN}^{\frac{1}{N}}\frac{1}{1-x}dx + o(k^{-1}) \\
&=N\left (-\log\left (1- \frac{1}{N}\right ) + \log\left (1-\frac{N_{n}}{kN}\right )\right )+ o(k^{-1})\\
\end{align*} 
in probability a.s. (conditioned on $N_{n}$). In the above equality, by substituting $N_{n}$ by $k/2$, we get
\begin{align*}
&N\left (-\log\left (1- \frac{1}{N}\right ) + \log\left (1-\frac{1}{2N}\right )\right )\\
&= N\left (-\log\left (\frac{N - 1}{N}\right ) + \log\left (\frac{2N - 1}{2N}\right )\right )\\
&= N \log\left (\frac{2N - 1}{2N}\frac{N}{N - 1} \right ) = \log\left (\frac{2N-1}{2N-2}\right ).
\end{align*}  Hence,
\begin{align*}
\lim_{k\rightarrow \infty}\mathbb{P}\left ((1-\epsilon)\frac{\Delta}{k} \leq N \log \frac{2N-1}{2N-2}\right ) &= \lim_{k\rightarrow \infty}\mathbb{P}\left (\frac{N_{n}}{k} \geq \frac{1}{2}\right ) \\
&= 1.
\end{align*}
The proof follows from this.
\end{proof}

\subsection{Proof of Theorem \ref{T:fluid bound}}
An ancillary result that we will use repeatedly without further explanation is the following.
Consider the differential equation
\[
\frac{d}{dt} x(t) =  -a(t) x(t) + b(t)
\]
with $x(0)= x_0 \ge 0,$ $a(t),b(t)$ Lipschitz continuous, $a(t)>0,b(t) \ge 0$ and $a(t)$ bounded above. Then $x(t) \ge 0.$ 
This follows from standard existence and uniqueness results on differential equations \cite{DIFFEQREF}
and on the solution using Duhamel’s principle:
\[
x(t) = x_0 e^{-z(t)} + \int_0^{z(t)} e^{-z(t)+s} \frac{b(z^{-1}(s))}{a(z^{-1}(s))} d s
\]
where $z(t) = \int_0^t a(u) du.$

\begin{lemma}[Monotonicity]\label{lem:monotonicity}
If $h_i(0) \ge \tilde{h}_i(0)$ for all $i$ then $h_i(t) \ge \tilde{h}_i(t)$
for all $i$ and $t.$
\end{lemma}
\begin{IEEEproof}
Assume $h_i(0) \ge \tilde{h}_i(0).$
It follows that $T_1 \le \tilde{T}_1$ and that $h_1(t) \ge \tilde{h}_1(t)$ for all $t.$
We proceed by induction.
Hence, assume that $h_j(t) \ge \tilde{h}_j(t)$ for all $t$ for some $j \ge 1.$
This implies that $T_{j} \le \tilde{T}_{j}.$
For $t <  T_{j}$ we have
by \eqref{eqn:hdiffeq} that
\begin{align*}
\frac{d}{dt} (h_{j+1}(t) - \tilde{h}_{j+1}(t)) =& -\frac{1}{\tilde{h}_{j}(t)}(h_{j+1}(t) - \tilde{h}_{j+1}(t))
\\&+\bigl(\frac{1}{\tilde{h}_{j}(t)} - \frac{1}{{h}_{j}(t)}\bigr) h_{j+1}(t).
\end{align*}
Hence we see that  $h_{j+1}(t) \ge \tilde{h}_{j+1}(t)$ for all $t < T_{j}.$
For $t \in (T_{j},\tilde{T}_{j})$ the above equation still holds with 
${h}_{j}(t) = N,$ hence we still have $h_{j+1}(t) \ge \tilde{h}_{j+1}(t).$
For $t > \tilde{T}_{j},$ $t< \min\{T_{j+1},\tilde{T}_{j+1}\}$ we have
\begin{align*}
\frac{d}{dt} (h_{j+1}(t) - \tilde{h}_{j+1}(t)) = -\frac{1}{N}(h_{j+1}(t) - \tilde{h}_{j+1}(t))
\end{align*}
and we again have $h_{j+1}(t) \ge \tilde{h}_{j+1}(t).$  Hence 
$T_{j+1} \le \tilde{T}_{j+1}$ and  $h_{j+1}(t) \ge \tilde{h}_{j+1}(t)$ for all $t.$
\end{IEEEproof}

Given initial conditions the system \eqref{eqn:hdiffeq}
can be solved as follows.
Define ${Q_0(t)}=e^{\frac{t}{N}},$ 
and for $i > 1$ define 
\[
{Q_i(t)} = \int_0^t {Q_{i-1}(z)} \,dz + h_i(0) {Q_{i-1}(0)}.
\]
(Note that in general we have $Q_i(0) = \prod_{j=1}^i h_j(0).$)
The solution to \eqref{eqn:hdiffeq} for $t \le T_1$ is then given by
\[
h_i(t) = \frac{{Q_i(t)}}{{Q_{i-1}(t)}}
\]
We can verify this directly:  First note that the initial conditions are satisfied and that the solution is correct
for $i=1.$ The key point is that $\frac{d}{dt} Q_i(t) = Q_{i-1}(t),$ so for $i>1$ we obtain
\begin{align*}
\frac{d}{dt}h_i(t) &= \frac{{Q_{i-1}(t)}}{{Q_{i-1}(t)}} - \frac{{Q_{i}(t)}Q_{i-2}(t)}{{Q^2_{i-1}(t)}}
\\&= 1 - \frac{h_i(t)}{h_{i-1}(t)}\,.
\end{align*}

Let us consider the system initialized with $h_i(0) = 0$ for $i \ge 1.$
For $t < T_1$ we then have
\begin{align*}
Q_i(t) &  = N^i\sum_{k=i}^\infty \frac{(t/N)^k}{k!} 
\end{align*}
which yields
\[
h_i (t) = \frac{t}{i}\, (1+O(t/N))\,.
\]

We now have two ways of analyzing the system, one using the differential equations
and another using the algebraic approach based on the above.
The differential equations are useful for establishing monotonicity properties of the
solution and the algebraic approach is useful for characterizing limiting behavior.

We will first consider the analysis of the differential equations
and focus primarily on the interval $[0,T_1].$
We assume initial conditions $h_i(0)$ that are non-increasing in $i$ and
satisfy $h_1(0)<1.$  Let us denote $T_1$ simply as $T.$
Define for $i \ge 1,$
\begin{align*}
r_i(t) = \frac{h_i(t)}{h_{i-1}(t)}; \,\, 
R_i(t) = \frac{1}{2-r_{i}(t)}; \,\, 
\alpha_i(t) = \frac{2-r_{i}(t)}{h_{i}(t)};
\end{align*}
then we have for $i>1,$
\begin{equation}\label{eqn:rdiffeq}
 \frac{d}{dt} r_i(t) = \alpha_{i-1}(t)(R_{i-1}(t)-r_i(t))
\end{equation}
Thus, we observe that $r_i(t)$ tracks $R_{i-1}(t).$

Let us call the sequence $h_i$ {\em regularly ordered} if
$h_i$ is strictly positive, monotonically decreasing with $h_1 \le 1,$  
$r_i \ge R_{i-1}$ for all $i \ge 2.$ 

\begin{lemma}\label{lem:decreasingr}
Assume a regularly ordered initial condition $h_i(0).$
Assume further that intializing with $h_i(0)$ gives
 $r_2(T) \ge \frac{1}{2-\frac{1}{N}}.$  
Then the solution has $r_i(t)$ non-increasing 
and $r_i(t) \ge R_{i-1}(t)$ for all $i \ge 2$ and $t \in [0,T].$
\end{lemma}
\begin{IEEEproof}
We first remark that each $\alpha_i$ is
finite and lies in a finite interval bounded away from $0.$
The proof is by induction on $i.$
First note that $R_1(t) = \frac{1}{2-\frac{h_1(t)}{N}}$ is non-decreasing and,
if $N< \infty$ it is increasing.
It therefore follows from \eqref{eqn:rdiffeq} that the assumed condition
$r_2(T) \ge \frac{1}{2-\frac{1}{N}} = R_1(T)$
implies that $r_2(t) \ge R_1(t)$ 
and that $r_2(t)$ is non-increasing for all $t \in [0,T].$
We proceed by induction.
Assume for some $i \ge 2$ that 
$r_{i}(t)$ is non-increasing and $r_{i}(t) \ge R_{i-1}(t)$ for $t\in [0,T].$ 
Then $R_i(t) = \frac{1}{2-r_i(t)}$ is non-increasing. 
Since $R_i(0) \le r_{i+1}(0)$ we conclude from \eqref{eqn:rdiffeq}
that $r_{i+1}(t)$ is non-increasing and $r_{i+1}(t) \ge R_{i}(t)$ for $t\in [0,T].$ 
\end{IEEEproof}
Note that the above lemma implies that if an initial condition $h_i(0)$ is regularly ordered
and satisfies $r_2(T) \ge R_1(T),$
then $h_i(t)$ is regularly ordered for all $t\in[0,T].$

\begin{lemma}\label{lem:monosys}
Assume two regularly ordered initial conditions $h'_i(0)$ and $h_i(0)$
where $h_1'(0) \ge h_1(0),$ and $r'_i(0) \ge r_i(0)$ for all $i \ge 2.$
Assume  further that $r_2(T) \ge \frac{1}{2-\frac{1}{N}}.$  
Then 
$h_i'(T') \ge h_i(T)$ and $r'_i(T') \ge r_i(T)$ for all $i \ge 1.$
\end{lemma}
\begin{IEEEproof}
By Lemma \ref{lem:decreasingr}, for $i \ge 2$ 
$r_i(t)$ is non-increasing on $[0,T]$ and
$r_i(t) \ge R_{i-1}(t).$ 

Note that $T' \le T$ and consider first $t\in [0,T'].$ 
We have $r'_1(t) \ge r_1(t)$ and $h'_1(t) \ge h_1(t)$ since $h'_1(0) \ge h_1(0).$
We proceed by induction.
Assume for some $i\ge 1$ that
$r'_i(t) \ge r_i(t)$  and $h'_i(t) \ge h_i(t).$ Then
$R'_i(t) \ge R_i(t)$ and $\alpha'_i(t) \le \alpha_i(t).$
From \eqref{eqn:rdiffeq} we have 
\begin{align*}
& \frac{d}{dt}(r'_{i+1}(t)-r_{i+1}(t)) 
\\&
= \alpha'_i(t)(R'_i(t)-r'_{i+1}(t)) -\alpha_i(t)(R_i(t)-r_{i+1}(t))
\\&
= -\alpha'_i(t)(r'_{i+1}(t)-r_{i+1}(t))
+ \alpha'_i(t)(R'_i(t)-R_i(t)) 
\\&\quad
+(\alpha_i(t)-\alpha'_i(t))(r_{i+1}(t)-R_i(t))
\,.
\end{align*}
By Lemma \ref{lem:decreasingr} the last term is non-negative and the second 
term is non-negative by the above argument.
Since $r'_{i+1}(0)-r_{i+1}(0)\ge 0$ the above equation clearly implies
 $r'_{i+1}(t)-r_{i+1}(t)\ge 0.$
Since $h'_i(t) \ge h_i(t)$ this implies $h'_{i+1}(t) \ge h_{i+1}(t).$
By induction we now have 
$r'_{i}(t) \ge r_{i}(t)$ on $[0,T']$ for all $i \ge 1.$

By Lemma \ref{lem:decreasingr}, $r_i(t)$ is decreasing on $[T',T]$ and we conclude that $r_i(T) \le r_i(T') \le r'_i(T').$
Since $h'_1(T')=h_1(T)=1$ we obtain $h'_i(T')=\prod_{j=2}^i r'_j(T') \ge \prod_{j=2}^i r_j(T) =h_i(T).$
\end{IEEEproof}

We will now introduce some additional notation to capture the renewal nature of the system.
Let us define $h_i^{[k]}(t) = h_i(T_k+t)$ and $T^{[k]} = T_k - T_{k-1}.$
Thus, we consider solving the system in a sequence of rounds and we
 use the superscript $\cdot^{[k]}$ to denote round $k.$
The initial condition for round $k$ is taken from the ending state of round $k-1.$
Hence $h^{[k]}_i(0) = h^{[k-1]}_{i+1}(T^{[k-1]})$ for $i \ge 1.$

Let us call a sequence $h_i(0)$ {\em fixed point convergent} if
$h_i(0)$ is regularly ordered and we have
$r_2(T) \ge \max\{ R_1(T),h_1(0) \},$ and
$r_i(T) \ge r_{i-1}(0)$  for  $i \ge 3.$

\begin{lemma}\label{lem:wellsetiterate}
If $h_i (0)$ is fixed point convergent  then
$h_i^{[k]}(0)$ is fixed point convergent for each $k$ and is monotonically
increasing $k$ converging to a fixed point.
\end{lemma}
\begin{IEEEproof}
The proof proceeds by induction.
Assume $h_i^{[k]}(0)$ is fixed point convergent.
Let us consider the interval $t \in [0,T^{[k]}].$
By Lemma \ref{lem:decreasingr},
\begin{equation}\label{eqn:rRineq}
r^{[k]}_{i}(t) \ge R^{[k]}_{i-1}(t)
\end{equation}
 for all $i \ge 2.$
Since  $\frac{1}{2-x} \ge x$ for $x \in [0,1],$ we have $R^{[k]}_{i-1}(t) \ge r^{[k]}_{i-1}(t)$
hence $r^{[k]}_{i}(t) \ge r^{[k]}_{i-1}(t)$ for all $i \ge 2.$

We now consider $h_i^{[k+1]}(0) = h^{[k]}_{i+1}(T)$ for $i \ge 1.$
Let us first show $h_i^{[k+1]}(0) \ge h_i^{[k]}(0)\,.$
We have 
\[
h_1^{[k+1]}(0) = h_2^{[k]}(T) = r_2^{[k]}(T) \ge h_1^{[k]}(0)
\]
where the last step is by assumption that $h_i^{[k]}$ is fixed point convergent.
For $i \ge 2$ we have
\[
r_i^{[k+1]}(0) = r_{i+1}^{[k]}(T^{[k]}) \ge r_i^{[k]}(0) 
\]
where again, the last step is by assumption.
We now obtain for $i > 1$
\begin{align*}
h_i^{[k+1]}(0) &= h_1^{[k+1]}(0) \prod_{j=2}^i  r_{j}^{[k+1]}(0) 
\\&\ge  h_1^{[k]}(0) \prod_{j=2}^i  r_{j}^{[k]}(0) 
= h_i^{[k]}(0)\,.
\end{align*}

Now we will show that $h_i^{[k+1]}(0)$ is also fixed point convergent.
Clearly, the sequence is monotonically decreasing and positive.
For $i \ge 3$ we have
\[
r_i^{[k+1]}(0) = r_{i+1}^{[k]}(T) \ge R_i^{[k]}(T) = R_{i-1}^{[k+1]}(0)\,
\]
where the middle inequality uses \eqref{eqn:rRineq}.
Since we also have
\(
r_2^{[k+1]}(0) = r_{3}^{[k]}(T) \ge R_2^{[k]}(T) > R_{1}^{[k+1]}(0)\,
\)
we see that $h_i^{[k+1]}(0)$ is regularly ordered.

Consider now $r_2^{[k+1]}(T^{[k+1]}).$
Let $\tau$ be determined by $h_1^{[k]}(\tau) = h_1^{[k+1]}(0).$
Then for $t \le T^{[k+1]}-\tau$ we have $h_1^{[k]}(\tau+t) = h_1^{[k+1]}(t)$
hence  $\tau+T^{[k]} = T^{[k+1]}$ and $R_1^{[k]}(\tau+t) = R_1^{[k+1]}(t).$
Since $r_2^{[k+1]}(0) \ge r_2^{[k]}(0) \ge r_2^{[k]}(\tau)$ we see that
\eqref{eqn:rdiffeq} now implies that $r_2^{[k+1]}(t) \ge r_2^{[k]}(\tau+t).$
Hence 
\[
r_2^{[k+1]}(T^{[k+1]}) \ge r_2^{[k]}(T^{[k]}) \ge  \frac{1}{2-\frac{1}{N}}
\]
and since $r_2^{[k]}(T^{[k]}) = h_1^{[k+1]}(0)$ we have
\[
r_2^{[k+1]}(T^{[k+1]})  \ge  \max\{ h_1^{[k+1]}(0),\frac{1}{2-\frac{1}{N}}\}.
\]

We now apply Lemma \ref{lem:monosys} to obtain  for $i \ge 3,$
\[
r_i^{[k+1]}(T^{[k+1]}) \ge r_i^{[k]}(T^{[k]}) = r_{i-1}^{[k+1]}(0)\,.
\]

We conclude that $h_i^{[k]}$ is a monotonically increasing sequence that therefore has a limit.
Correspondingly, $T^{[k]}$ is a montonically decreasing sequence that also has limit.
The limit is a fixed point by continuity.
\end{IEEEproof}

\subsection{Fixed Points}
Simulations show that $h_i^{[k]}(0)$ converges in $k$ to a fixed point solution.
In this section we solve for the set of fixed points.

In general for $i \ge 1$ we have
$h^{[k]}_i (0)  = \frac{Q^{[k-1]}_{i+1}(T^{[k-1]})}{Q^{[k-1]}_{i}(T^{[k-1]})}.$
Since $Q^{[k]}_0(0) = 1$ for all $k$ we have
$Q^{[k]}_1(0) = h^{[k]}_1(0)  = e^{-\frac{T^{[k-1]}}{N}} Q^{[k-1]}_2(T^{[k-1]}).$
Since 
\(
Q^{[k]}_i(0) = h^{[k]}_i(0) Q^{[k]}_{i-1}(0)
\)
we  proceed inductively to obtain
\[
Q^{[k]}_i(0)
= e^{-\frac{T^{[k-1]}}{N}} Q^{[k-1]}_{i+1}(T^{[k-1]}).
\]

Hence, we can express $Q^{[k]}(t)$ directly in terms of $Q^{[k-1]}(T^{[k-1]})$ 
as follows.
\begin{align*}
Q^{[k]}_i(t) &
=e^{-\frac{T^{[k-1]}}{N}} \sum_{j=0}^{i-1}  Q^{[k-1]}_{i+1-j}(T^{[k-1]})\frac{1}{j!}t^j
+   N^i \sum_{j=i}^{\infty}  \frac{(t/N)^j}{j!}
\\&
=\sum_{j=0}^{i-1}  Q^{[k]}_{i-j}(0)\frac{1}{j!}t^j
+   N^i \sum_{j=i}^{\infty}  \frac{(t/N)^j}{j!}\,.
\end{align*}

Noting at the fixed point we have
$Q_{i}(T) = e^{\frac{T}{N}} Q_{i-1}(0)$
we obtain
\begin{align*}
e^{\frac{T}{N}}Q_{i-1}(0) &
= \sum_{j=0}^{i-1}  Q_{i-j}(0)\frac{1}{j!}T^j
+   N^i \sum_{j=i}^{\infty}  \frac{(T/N)^j}{j!}
\\
&
= \sum_{j=0}^{i}  Q_{i-j}(0)\frac{1}{j!}T^j
+   N^i \sum_{j=i+1}^{\infty}  \frac{(T/N)^j}{j!}
\end{align*}
Dividing through by $Q_{i-1}(0)$ and rearranging terms we obtain
\begin{align*}
h_i(0) &
= e^{\frac{T}{N}} - T - T \sum_{j=2}^{i} \frac{1}{j!}
\prod_{k=1}^{j-1}\frac{T}{h_{i-k}(0)}
\\&-  \frac{1}{Q_{i-1}(0)} N^i \sum_{j=i+1}^{\infty}  \frac{(T/N)^j}{j!}
\end{align*}
Define $h_i(0) = N$ for $i \le 0$ we can write this as
\begin{align*}
h_i(0) &
= e^{\frac{T}{N}} - T - T \sum_{j=2}^{\infty} \frac{1}{j!}
\prod_{k=1}^{j-1}\frac{T}{h_{i-k}(0)}
\end{align*}
Note that this confirms $h_1(0)= e^{\frac{T}{N}} - N(e^{\frac{T}{N}}-1).$
It follows that $h_i(0)$ is a monotonically decreasing in $i$
and therefore has a limit $h_\infty.$
For $T=0$ we obtain $h_i(0) = 1$ for all $i$ and 
if $T>0$ then we must have $h_\infty > 0.$
Taking limits we obtain
\begin{align*}
h_\infty &
= e^{\frac{T}{N}} - T - T \sum_{j=2}^{\infty} \frac{1}{j!}
\prod_{k=1}^{j-1}\frac{T}{h_\infty}
 &
= e^{\frac{T}{N}} - h_\infty (e^{\frac{T}{h_\infty}} - 1)
\end{align*}
or
\begin{equation}\label{eqn:Thform}
h_\infty  = e^{-T(\frac{1}{h_\infty} -\frac{1}{N})},\quad
T = -\frac{h_\infty\ln h_\infty}{1-\frac{h_\infty}{N}} \,.
\end{equation}
Now, $-\frac{h_\infty\ln h_\infty}{1-\frac{h_\infty}{N}}$ is a concave function
of $h_\infty$ on $[0,1].$ Let $\Tfixed(N)$ denote its  maximum.
We have $\Tfixed(\infty) = e^{-1}.$

\begin{lemma}
For each $T \le \Tfixed(N)$ there exists a fixed point with the given $T.$
\end{lemma}
\begin{IEEEproof}
We need only show that the recursion for $h_i(0)$ is well behaved.
Let $T$ be as above and assume $h_j(0) \ge h_\infty = h_\infty(T).$
Then we have
\begin{align*}
h_i(0) &
= e^{\frac{T}{N}} - T - T \sum_{j=2}^{\infty} \frac{1}{j!}
\prod_{k=1}^{j-1}\frac{T}{h_{i-k}(0)}
\\& \ge e^{\frac{T}{N}} - T - T \sum_{j=2}^{\infty} \frac{1}{j!}
\prod_{k=1}^{j-1}\frac{T}{h_\infty}
\\& \ge e^{\frac{T}{N}} - h_\infty (e^\frac{T}{h_\infty} -1)
\\& =  h_\infty
\end{align*}
Since $h_i = N$ for  $i \le 0$ we see that the sequence is well defined.
\end{IEEEproof}

Let us denote the fixed point corresponding to $\Tfixed(N)$ as $\hfixed = \hfixed(N).$
Since $\hfixed_\infty(N) = \text{argmax} \frac{-h \ln h}{1-\frac{h}{N}}$ a little calculus shows that
$\hfixed_\infty(N)$ is determined by $-\ln \hfixed_\infty(N) = 1-\frac{h_\infty^*(N)}{N}$ or
$e \hfixed_\infty(N) = e^{\frac{\hfixed_\infty(N)}{N}} \ge 1 + \frac{\hfixed_\infty(N)}{N}$
and we obtain the bound 
\begin{equation}\label{eqn:hstarbound}
\hfixed_\infty(N) \ge \frac{N}{Ne-1}.
\end{equation}
Moreover we have from \eqref{eqn:Thform}
\begin{equation}\label{eqn:Tequality}
\Tfixed(N) =  \frac{-\hfixed_\infty(N) \ln \hfixed_\infty(N)}{1-\frac{\hfixed_\infty(N)}{N}} = \hfixed_\infty(N).
\end{equation}
From \eqref{eqn:Thform} we obtain
\begin{equation}\label{eqn:h1star}
\hfixed_1 = \hfixed_\infty e^{\Tfixed/\hfixed_\infty} - \hfixed_\infty = (e-1)\hfixed_\infty 
\end{equation}

By Lemma \ref{lem:monotonicity} the fixed point $\hfixed(N)$ gives an upper bound
on the solution for $N.$  This shows that the asymptotic gain in 
the transmission velocity of the file is upper bounded by that
from $\hfixed(N).$  In the case $N=\infty$ we prove in the next section
that the solution converges to $\hfixed(\infty).$
Simulations indicate that convergence to $\hfixed(N)$ occurs for all $N >1$
but we  do not have a proof for the genereal case.
We can show, however,
that the asymptotic acceleration is that determined by the fixed point.

\subsection{Convergence}

The case $N=\infty$ is the Fountain code case (infinitely low rate).
Let us initialize in round $1$ with $h^{[1]}_i(0)=0$ for $i>0.$ 
It is immediate that 
\[
Q^{[1]}_i(t) = \frac{1}{i!}t^i;\,\,\,
h^{[1]}_i(t) = \frac{1}{i} t;\,\,\,
r^{[1]}_i(t) = \frac{i-1}{i}.\,
\]
Hence, the case $k=1$ is fixed point convergent in that $h_i^{[1]}(t) =  \frac{t}{i}$
and $r_i(t)=R_{i-1}(t)$ for all $t.$
Hence the solution in round $1$ is fixed point convergent.

\begin{lemma}\label{lem:infiniteconverge}
In the $N=\infty$ case the sequence $h^{[k]}$ is monotonically increasing and approaches the limit $\hfixed.$
\end{lemma}
\begin{IEEEproof}
Since the solution above in round $1$ is fixed point convergent in the $N=\infty$
case, convergence of $h^{[k]}$ follows directly 
from Lemma \ref{lem:wellsetiterate}.  Since $T^{[k]}$ is decreasing
it has a limit and it follows from Lemma \ref{lem:monosys}
that this limit must be $\Tfixed(\infty).$
Indeed, if the limit were less than $\Tfixed(\infty)$ then we would obtain
a contraction to Lemma \ref{lem:monosys} by initializing the system
with $\hfixed(\infty)$ and comparing the $0$ initialization.
Hence the limit of $h^{[k]}$ is $\hfixed(\infty).$
\end{IEEEproof}

In the case of finite $N$ we see that $r_i^{[1]}(t)$ is increasing and so
the above monotonicity argument does not succeed.
Simulations indicate that in subsequent rounds the sequence does become fixed point convergent, but
we have not been able to prove this.
Thus we are unable to prove convergence of  $h^{[k]}$
although we conjecture that it converges to $\hfixed(N).$

The upper bound in Theorem \ref{T:fluid bound} is obtained by initializing with the fixed point $\hfixed(N),$ which
obviously has the stated asymptotic delay.
The lower bound is obtained by showing that the solution eventually exceeds a fixed point convergent condition.
Then Lemma \ref{lem:wellsetiterate} provides the lower bound.

Before  giving the proof for the lower bound we develop some more preliminary results.
\begin{lemma}\label{lem:decreasinginN}
For a fixed initial condition we can parameterize the solution on $[0,T]$ by $h_1.$
Then $h_i( h_1;N)$ is a decreasing function of $N.$
\end{lemma}
\begin{IEEEproof}
We consider parameterizing the solution on $[0,T]$
by $h_1$ which spans the interval $[h_1(0),1].$ Then we have $h_2(h_1)$ satisfies
\[
\frac{d}{dh_1} h_2 = \Bigl(1-\frac{h_2}{h_1}\Bigr) \Bigl(\frac{1}{1-\frac{h_1}{N}}\Bigr)
\]
from which it easily follows that $h_2(h_1)$ is a decreasing function of $N.$
Now assume that  $h_{i-1}(h_1)$ is a decreasing function of $N.$  Since we have
\[
\frac{d}{dh_1} h_i = \Bigl(1-\frac{h_i}{h_{i-1}}\Bigr) \Bigl(\frac{1}{1-\frac{h_1}{N}}\Bigr)
\]
it is easily seen that $h_i(h_1)$ is a decreasing function of $N.$
Hence the lemma follows by induction.
\end{IEEEproof}

\begin{lemma}\label{lem:Nspeedup}
For any $N \ge N'>1$ we have
$h(\gamma t;N') \ge h( t;N)$ where
$\gamma = \frac{1-\frac{1}{N}}{1-\frac{1}{N'}}.$
\end{lemma}
\begin{IEEEproof}
By Lemma \ref{lem:decreasinginN} we see that for round $[1]$ we
have $h_i(h_1;N') \ge h_i(h_1;N)$  where we have parameterized by $h_1$ instead of $t$
and indicated explicit dependence on $N.$  

Let us define $\gamma(x) = \frac{1-\frac{x}{N}}{1-\frac{x}{N'}}$
so the above $\gamma$ is $\gamma(1).$  Note that $\gamma(x)$ is increasing in $x$ on $[0,1].$
Now consider a fixed initial condition.  Then for $\gamma t  \le T(N')$ and $t \le T(N)$ we have
\begin{align*}
\frac{d}{dt} h_1(\gamma t;N') &=  \gamma (1-\frac{h_1(\gamma t;N')}{N'}) \\
& =  \gamma (h_1(\gamma t;N')) (1-\frac{h_1(\gamma t;N')}{N'}) \\&\quad  + (\gamma(1)-\gamma (h_1(\gamma t;N')))(1-\frac{h_1(\gamma t;N')}{N'}) \\
& =   (1-\frac{h_1(\gamma t;N')}{N}) \\&\quad + (\gamma(1)-\gamma (h_1(\gamma t;N')))(1-\frac{h_1(\gamma t;N')}{N'}) \\
\end{align*}
and we obtain
\begin{align*}
\frac{d}{dt} (h_1(\gamma t;N') &-h_1(t;N)) = -\frac{1}{N}(h_1(\gamma t;N') -h_1(t;N))  \\
&  + (\gamma(1)-\gamma (h_1(\gamma t;N')))(1-\frac{h_1(\gamma t;N')}{N'}) \\
\end{align*}
hence $h_1(\gamma t;N') \ge h_1(t;N).$  It follows that $T(N) \le  \gamma^{-1}T(N').$
By Lemma \ref{lem:decreasinginN} we have $h_i(\gamma t;N') \ge h_i(t;N)$ for all $i$ and
also $h(T(N');N') \ge h(T(N),N).$ We can now conclude that
$h(\gamma t;N') \ge h(t;N)$ for all $t \in [0,T(N)].$

Now consider the system initialized  with $h=0.$  By the above we have
$h(\gamma t;N') \ge h(t;N)$ for $t  \le  T_1(N).$
If we reduce $h(\gamma T_1(N);N')$ to set it equal to $h( T_1(N);N)$ then the above argument
would again yield  $h(\gamma t ;N')\ge h(t;N)$ for $t  \in  [T_1(N),T_2(N)].$
Now, Lemma \ref{lem:monotonicity} implies that without the reduction $h(\gamma t ;N')$ would be larger still
so we have $h(\gamma t ;N')\ge h(t;N)$ for $t  \in  [T_1(N),T_2(N)]$ for the actual solution.
The same argument can be repeated for $t\in[T_i(N),T_{i+1}(N)]$ for $i=2,3,...$ and hence by induction we obtain
$h(\gamma t ;N')\ge h(t;N)$ for all $t.$
\end{IEEEproof}

Let us define for $i\ge 1$ and $N \in (1,\infty],$
\[
h^*_i(N)=\frac{1}{(i+1)-\frac{i}{N}}\text{ and } r^*_{i}(N) = \frac{h^*_i}{h^*_{i-1}}
\]
with $h^*_0(N)=N.$
The key property of this definition is that for $i \ge 2$ we have
\begin{equation}\label{eqn:rstareq}
r^*_i(N) = \frac{1}{2-r^*_{i-1}(N)}
\end{equation}

\begin{lemma}\label{lem:Rinit}
The initial condition $h^*(M)$
is fixed point convergent
for a given $N \le M$ if, given the initial condition, we have
$h_2(T) \ge h^*_1(N).$
Moreover, we then have $r_i(T) \ge r_{i-1}^*(N),$ for all $i \ge 3,$
hence $h_i(T) \ge h^*_{i-1}(N),$ for all $i \ge 3.$
\end{lemma}
\begin{IEEEproof}
It follows from \eqref{eqn:rstareq} that $h_i^*(M)$ is regularly ordered.
To prove it is fixed point convergent for $N$ we need to show that
$r_2(T) \ge \max\{ R_1(T),h_1(0) \},$ and that
$r_i(T) \ge r_{i-1}(0)$  for  $i \ge 3.$
We assume that $r_2(T) = h_2(T) \ge \frac{1}{2-\frac{1}{N}} = R_1(T) = h_1(0)$
so to prove it is fixed point convergent for $N$ we need only show that
$r_i(T) \ge r_{i-1}(0)$  for  $i \ge 3.$
Since $r_{i-1}(0) =r^*_{i-1}(M) \le r^*_{i-1}(N)$ we see that it
is sufficient to show $r_i(T) \ge r_{i-1}^*(N)$ for all $i \ge 3.$

We have by assumption that $r_2(T) \ge r_1^*(N).$
We proceed by induction.
Hence assume that $r_i(T) \ge r^*_{i-1}(N).$ 
It follows from Lemma \ref{lem:decreasingr} that
$r_i(t)$ is decreasing and $r_i(t) \ge R_{i-1}(t)$ for all $i \ge 2.$
Hence  $r_{i+1}(T) \ge R_i(T) = \frac{1}{2-r_i(T)} \ge \frac{1}{2-r^*_{i-1}(N)} = r^*_{i}(N).$
\end{IEEEproof}

\begin{lemma}\label{lem:NMbounds}
There exists $\delta>0$ such that if $N\in [2,6]$ then 
the initial condition $h_i^*(M)$ 
is fixed point convergent for $N$ for all $M \in [N,N+\delta].$
If $N \ge 6$ then the initial condition $h_i^*(\infty)$ is fixed point convergent.
\end{lemma}
\begin{proof}
By Lemma \ref{lem:Rinit} $h^*_i(M)$ is fixed point convergent for $N$
if $h_2(T) \ge h^*_1(N).$
We can obtain conditions for this using the solution obtained in terms of the functions $Q_i(t).$
In particular we have
\begin{align*}
Q_0(T)&=e^{\frac{T}{N}} 
\\
Q_1(T)&=h^*_1(M) + N(e^{\frac{T}{N}}-1)
\\
Q_2(T)&=h^*_1(M) h^*_2(M) + h^*_1(M) T+ N^2(e^{\frac{T}{N}}-1-\frac{T}{N})
\end{align*}
and $T = T(N)$ is determined by $Q_1(T) = Q_0(T).$
We have $h_2(T) = Q_2(T)/Q_1(T)$ so the initial condition is fixed point convergent
if 
\[
Q_2(T) \ge Q_1(T)  h_1^*(N)
\]
under the condition $Q_1(T) = Q_0(T),$ which determines $T.$

The condition $Q_1(T) = Q_0(T)$ gives 
\[
Q_0(T)=Q_1(T) = e^{\frac{T}{N}}=\frac{N-h^*_1(M))}{N-1} = 1+\frac{1-h^*_1(M)}{N-1}
\]
hence
\[
T = N \ln \bigl(1+\frac{1-h^*_1(M)}{N-1} \bigr)
\]

Let us first consider the case $M=\infty.$ We have $h^*_1(M)=1/2$ and $h^*_2(M)=1/3$ and
$e^\frac{T}{N} = \frac{2N-1}{2(N-1)}.$  Using $N^2(e^{\frac{T}{N}}-1-\frac{T}{N}) \ge \frac{1}{2}T^2$
we see that $h_i^*(\infty)$ is fixed point convergent for $N$ if
\begin{align*}
\frac{1}{6}+\frac{1}{2}T +\frac{1}{2}T^2  \ge \frac{N}{2(N-1)} = \frac{1}{2}+\frac{1}{2N-2}
\end{align*}
Some algebra shows that $T(N)+T(N)^2-\frac{1}{N-1}$ is increasing in $N$ on $[2,\infty)$
(although we only use this for $N\ge 6$)
and we see that if the equality holds for $N=N'$ then it holds
for all $N \ge N'.$
It can be easily verified that the inequality holds for $N=6.$

It is clear that $h_1(0)$ and $h_2(0)$ are uniformly continuous in $M$ on $[2,6].$ 
Thus, to obtain the desired result we need only verify that $Q_2(T) > e^{\frac{T}{N}} \frac{1}{2-\frac{1}{N}}$
when $M=N$ for $N \in [2,6].$ Noting that for $N=M$ we have $h_1(0) = \frac{1}{2-\frac{1}{N}}$ we can, by dividing through
by $h_1(0)$ and rearranging terms, write the condition to be shown as as
\begin{align*}
  (1-\frac{1}{N}) T
+ \biggl(2-\frac{1}{N}-\frac{1}{N^2}\biggr) \sum_{k=2}^\infty \frac{1}{N^{k-2}}\frac{T^k}{k!}
> \frac{2N-2}{3N-2}
\end{align*}
where in this case we have $T(N) = N\ln (1+ \frac{1}{2N-1}).$  This inequality holds on $[2,6]$ taking only the $k=2$ and
$k=3$ terms from the sum, which can be verified with some algebra using the bound for $N \ge 2,$
\[
T(N) \ge N \Bigl(\frac{1}{2N-1} - \frac{1}{2(2N-1)^2} + \frac{1}{4(2N-1)^3}\Bigr)\,.
\] 
\end{proof}

\begin{IEEEproof}[Main Proof of Theorem \ref{T:fluid bound} - lower bound]
We assume $N\ge 2.$

We can obtain a lower bound by showing that for some $t=\tau$ we have
$h(\tau;N) \ge \hat{h}$ where $\hat{h}$ is fixed point convergent (for the given $N.$)
To see why this produces a lower bound with the desired property consider
initializing the system with $\hat{h}.$ By Lemma \ref{lem:wellsetiterate} the resulting solution 
$\hat{h}(t;N)$ then satisfies $\lim_{k\rightarrow \infty} \hat{h}^{[k]}(N) = \hfixed(N)$
with $\hat{T}^{[k]}$ a decreasing sequence that approaches $\Tfixed(N).$
It follows that for any $\epsilon >0$ we have $\hat{T}_k \le k (\Tfixed(N)+\epsilon)$
for $k$ large enough.
Now, by Lemma \ref{lem:monotonic}, $\hat{h}(t;N) \le h(\tau+t;N)$ so $T_k \le \tau+\hat{T}_k.$
Since $\epsilon$ is arbitrary we now obtain 
$\limsup_{k\rightarrow \infty} \frac{1}{k} T_k \le \Tfixed(N)$
which together with the lower bound gives the result.

By Lemma \ref{lem:Nspeedup} we see that it is sufficient to find $\hat{h}(t,N')$ that is fixed point convergent for $N$
for any $N' \ge N.$
For $N\ge 6$ we have $h_i(T_1(\infty);\infty) = \frac{1}{i}$ is fixed point convergent by Lemma \ref{lem:NMbounds}.
Since $h_i(T_1(N);N) \le h_i(T_1(\infty);\infty)$ for $N<\infty$
by Lemma \ref{lem:decreasinginN} the proof is complete for $N\ge 6.$

From Lemma \ref{lem:Rinit} it follows that for some $\tau$ we have $h_i(\tau;N=6) \ge h^*_i(6).$
This gives a fixed point convergent sequence for $N\in [6-\delta,6].$
Similarly, it now follows from Lemma \ref{lem:Rinit} that for some $\tau$ we have $h_i(\tau;N=6-\delta) \ge h^{*}_i(6-\delta).$
This gives a fixed point convergent sequence for $N\in [6-2\delta,6-\delta].$
Proceeding by induction we obtain fixed point convergent conditions for all $N\ge 2.$
%
%
%

We conclude by noting that the lower bound and upper bound are asymptotically equal.
\end{IEEEproof}
\subsection{Sketch of Proof of Proposition\ref{P: proportioned fluid bound}}
As before there occurs a sequence of times $T_1,T_2,...$
where $T_i$ denotes the time $t$ where $h_i(t)$ reaches $1.$
For $t \le T_1$ we have
\[
h_1(t) = h_1(0)+t
\]
We obtain for $i>1,$
\begin{equation}\label{eqn:hPRdiffeq}
e^t h_i(t) = \int_0^t e^t h_{i-1}(s) ds +  h_i(0)\,.
\end{equation}

Defining $r_i := \frac{h_i}{h_{i-1}}$ for $i \ge 2$ we get
\begin{equation}\label{eqn:revolvePR}
\frac{d}{dt} r_i(t) = 1 - \frac{r_i(t)}{r_{i-1}(t)}\,.
\end{equation}
where we also define $r_1(t) := \frac{h_1(t)}{h_1(t)+1}.$

Let us consider the initialization $h_i(0)=0, i\ge 1$ and the interval $[0,T_1].$
It is easy to see from \ref{eqn:hPRdiffeq} that $r_i(t)$ approaches $0$
as $t\rightarrow 0.$

\begin{lemma}
If $r_i(0) \le \frac{1}{2}$ for $i \ge 2$ then
we have $r_i(t) \le \frac{1}{2}$ for  $i \ge 1$ and  $t \in [0,T].$
\end{lemma}
\begin{IEEEproof}
Consider any initial condition $r_i(0) \le \frac{1}{2}.$
Then on $[0,T]$ we have $r_1(t) \le \frac{1}{2}$ and so
\eqref{eqn:revolvePR} implies
$r_i(t) \le \frac{1}{2}$ for $i=2,3,...$ by induction.
\end{IEEEproof}

The above Lemma shows that $h_i$ decays rapidly in $i.$ This also implies that the 
probability that a transmission will be successful is at least $\frac{1}{2}$  so
the ratio of wasted to useful transmissions is finite.

Given an initial condition we can obtain a solution as
\[
h_1(t) = h_1(0)+t
\]
and for $i \ge 1,$
\[
e^{t}h_i (t) =\int_0^t e^s h_{i-1}(s) ds + h_i(0)\,.
\]
Even though we have a simple recursive form it appears difficult to prove convergence.
Simulations indicate that $h^{[k]}$ is an increasing sequence that therefore converges.

\subsection{Fixed Points}\label{A:proportion forwarding}

Consider $M_k(t) = \sum_{i=k}^\infty h_i(t).$
Since $h_i(t)$ goes to $0$ in $i$ we have for $k \ge 2,$
\[
\frac{d}{dt}M_k = h_{k-1}
\]
and at a fixed point we have
\[
\int_0^T h_{k-1}(s) ds = M_k(T) - M_k(0) = M_{k-1}(0) - M_k(0) = h_{k-1}(0).
\]
In particular setting $k=2$ we obtain gives the necessary condition
\[
\int_0^T h_{1}(s) ds = h_{1}(0) T + \frac{1}{2}T^2 = h_{1}(0)
\]
and since $h_1(0) = 1-T$ we can solve to obtain $T = \sqrt{2}.$
Since this determines $h_1(t)$ we can then in principle solve for
$h_i(0), i=2,3,...$ using
\[
\int_0^T h_{i}(s) ds = h_{i}(0) 
\]
and \eqref{eqn:hPRdiffeq}.

Analysis of the long term evolution of $M_1^{[k]}(0)$ shows that the asymptotic speed up is at least $5/3.$

In general we have
\begin{align*}
M_1^{[k+1]}(0) & = M_1^{[k ]}(T^{[k]}) - 1 \\
 & = M_1^{[k]}(0) + \int_0^{T^{[k]}} h^{[k]}_1(s) ds + (1-h^{[k]}_1(0))- 1 \\ 
 & = M_1^{[k]}(0) + 2T^{[k]} -\frac{1}{2}(T^{[k]})^2 - 1
\end{align*}
Let us define $p(T) = 2T -\frac{1}{2}(T)^2 - 1.$
Since $0 \le M_1^{[k]}(0) \le 1$ we easily obtain
$|\sum_{k=1}^i p (T^{[k]}) | \le 1.$
We have $T^{[k]} \ge \frac{1}{2}.$
Since $p(T) = 2T -\frac{1}{2}(T)^2 - 1$ is concave increasing on $[\frac{1}{2},1]$ 
we see that for any probability distribution of $T$ on  $[\frac{1}{2},1]$ 
the point $(\expectation (T),\expectation (P(T)))$ is below the graph of $p(T)$ and
above the line segment joining the endpoints.
Since $p(\frac{1}{2}) = -\frac{1}{8}$ and $p(1) = \frac{1}{2}$ 
any distribution of $T$ with $\expectation(P(T)) =0$
has $\expectation{(T)} \le \frac{3}{5}.$
Hence,
\begin{align*}
\limsup_{i \rightarrow \infty} \frac{1}{i}T_i  = 
\limsup_{i \rightarrow \infty} \frac{1}{i}\sum_{k=1}^i T^{[k]}  \le \frac{3}{5}.
\end{align*}
We also have the corresponding lower bound (which we conjecture is tight) of
\begin{align*}
\liminf_{i \rightarrow \infty} \frac{1}{i}\sum_{k=1}^i T^{[k]}  \ge 2-\sqrt{2} \simeq  0.5857...
\end{align*}
 
\end{document}